\documentclass[a4paper,11pt]{amsart}
\usepackage{amsmath,amssymb,amsthm}
\usepackage{epsfig}
\usepackage{graphicx}

\usepackage[hmargin=2cm,vmargin=2.5cm]{geometry}

\usepackage{tikz-cd}

\makeatletter
\def\Ginclude@eps#1{%
 \message{<#1>}%
  \bgroup
  \def\@tempa{!}%
  \dimen@\Gin@req@width
  \dimen@ii.1bp%
  \divide\dimen@\dimen@ii
  \@tempdima\Gin@req@height
  \divide\@tempdima\dimen@ii
    \includegraphics{#1}%
  \egroup}
\makeatother
\renewcommand{\>}{\rangle}

\newcommand{\C}{\operatorname{\mathfrak{C}}}

\newcommand{\CP}{\mathbb{C}\mathbb{P}^1}

\newcommand{\Coef}{\operatorname{Coef}}

\newcommand{\Z}{\mathbb Z}
\newcommand{\mbC}{\mathbb C}
\newcommand{\mbP}{\mathbb P}

\newcommand{\mbZ}{\mathbb Z}

\newcommand{\eps}{\varepsilon}

\def\d{\partial}
\def\CP1{\mathbb{C}\mathbb{P}^1}

\newcommand{\oM}{\overline{\mathcal M}}

\newcommand{\og}{\overline g}
\newcommand{\oh}{\overline h}

\newcommand{\oG}{\overline G}

\def\oM{{\overline{\mathcal{M}}}}

\def\g{{\mathcal{G}}}

\def\Z{{\mathbb Z}}

\def\C{{\mathbb C}}
\def\Q{{\mathbb Q}}
\def\d{{\partial}}

\newcommand{\tC}{\widetilde C}
\newcommand{\Li}{\mathrm{Li}}
\newcommand{\mcS}{\mathcal S}
\newcommand{\tGamma}{\widetilde\Gamma}
\newcommand{\tgamma}{\widetilde\gamma}
\newcommand{\tp}{\widetilde p}

\newtheorem{theorem}{Theorem}[section]
\newtheorem{proposition}[theorem]{Proposition}

\newtheorem{lemma}[theorem]{Lemma}

\newtheorem{rem}[theorem]{Remark}
\newtheorem{ex}[theorem]{Example}
\newenvironment{remark}{\begin{rem}\rm}{\qee\end{rem}}

\newcommand{\qee}{\mbox{\hspace{0.2mm}}\hfill$\triangle$}

\numberwithin{equation}{section}

\title{Double ramification cycles and quantum integrable systems}

\author{Alexandr Buryak}
\address{A.~Buryak:\newline Department of Mathematics, ETH Zurich, \newline Ramistrasse 101 8092, HG G 27.1, Zurich, Switzerland}
\email{buryaksh\_at\_gmail.com}

\author{Paolo Rossi}
\address{P.~Rossi:\newline IMB, UMR 5584 CNRS, Universit\'e de Bourgogne,\newline 9, avenue Alain Savary, 21078 Dijon Cedex, France}
\email{paolo.rossi\_at\_u-bourgogne.fr}

\begin{document}

\begin{abstract}
In this paper we define a quantization of the Double Ramification Hierarchies of \cite{Bur14} and \cite{BR14}, using intersection numbers of the double ramification cycle, the full Chern class of the Hodge bundle and psi-classes with a given cohomological field theory. We provide effective recursion formulae which determine the full quantum hierarchy starting from just one Hamiltonian, the one associated with the first descendant of the unit of the cohomological field theory only. We study various examples which provide, in very explicit form, new $(1+1)$-dimensional integrable quantum field theories whose classical limits are well-known integrable hierarchies such as KdV, Intermediate Long Wave, Extended Toda, etc. Finally we prove polynomiality in the ramification multiplicities of the integral of any tautological class over the double ramification cycle.
\end{abstract}

\maketitle

\tableofcontents

\section*{Introduction}

One of the main features of the algebraic setting of Symplectic Field Theory (SFT) \cite{EGH00} and its approach to the relation between integrable systems and moduli spaces of holomorphic curves is the appearance of infinite dimensional quantum integrable systems associated to higher genus curves, as opposed to the fact that, in the Dubrovin-Zhang construction \cite{DZ05} of integrable hierarchies from cohomological field theories, higher genus curves control the dispersive expansion of a still classical system.\\

Of course the target manifolds considered by the two theories are different. In the simplest algebraic setting, SFT coincides with relative Gromov-Witten theory of a trivial $\mbP^1$-bundle over, say, a closed K\"ahler manifold $X$, relative to the zero and infinity sections. The Dubrovin-Zhang (DZ) hierarchy is instead associated with Gromov-Witten theory of the base manifold $X$.  At genus $0$ the two theories recover the same classical dispersionless integrable hierarchy, a system of conservation laws (basically because on a genus $0$ curve with marked points, there is a, unique up to a $\mbC^*$-symmetry, meromorphic function with given divisor of zeros and poles supported on the marked points). In higher genus, however, the two theories differ, the first giving rise to a quantization and the second to a dispersive expansion of the starting dispersionless system.\\

More in general one can substitute the Gromov-Witten theory of the target space $X$ with a cohomological field theory (CohFT) on $\oM_{g,n}$ and the relative Gromov-Witten theory of the $\mbP^1$-bundle with its intersection theory with the double ramification cycle $DR_g(a_1,\ldots,a_n) \in H^{2g}(\oM_{g,n};\mathbb{Q})$. In this setting, in a recent paper \cite{Bur14}, the first author made a conjecture that the right classical hierarchy in the SFT construction to be compared with the full genus dispersive DZ-hierarchy is not the genus~$0$ SFT, but instead the \emph{double ramification (DR) hierarchy}, corresponding to intersecting the given CohFT with $DR_g(a_1,\ldots,a_n) \times \lambda_g$, where $\lambda_i$ is the $i$-th Chern class of the Hodge bundle on $\oM_{g,n}$. The conjecture states that the DR-hierarchy and the DZ-hierarchy associated to a given CohFT are equivalent through a Miura tranformation (see \cite{DZ05} for details) and has been checked in \cite{Bur14} and \cite{BR14} for various CohFTs (trivial CohFT, Hodge CohFT, Gromov-Witten theory of $\mbP^1$ and in part for Witten's $r$-spin classes). This offers a natural candidate for the construction of a quantization of the full dispersive DZ-hierarchy, something we call the \emph{quantum double ramification (qDR) hierarchy}: the intersection of the given CohFT with $DR_g(a_1,\ldots,a_n) \times \Lambda(\epsilon)$, where this time $\Lambda(\epsilon)=1+\sum_{i=1}^g \epsilon^i \lambda_i$.\\

In particular, if we denote by $\eps$ and $\hbar$ the dispersion and quantization parameters respectively and put $\epsilon=\frac{-\eps^2}{\hbar}$, we can summarize the situation in the following diagram:

\[
\begin{tikzcd}
\boxed{\text{DZ}_{\eps=0}=\text{DR}_{\eps=0}} \arrow{r}{\hbar} \arrow[swap]{d}{\eps} & \boxed{\text{qDR}_{\eps=0}} \arrow{d}{\eps} \\
\boxed{\text{DZ}\simeq\text{DR}} \arrow{r}{\hbar} & \boxed{\text{qDR}}
\end{tikzcd}
\]
As explained above, the equivalence in the lower left corner of the diagram is via a Miura transformation and is  still conjectural in general (although the accumulating evidence is quite strong).\\

The qDR hierarchy has the nature of an integrable quantum field theory in one space and one time dimensions. It contains $N$ bosonic fields $u^\alpha(x,t) =\sum_{k \in \Z} p^\alpha_k e^{i k x}$, $\alpha=1,\ldots,N$ defined (formally) on the circle, where the Fourier coefficients $p_k^\alpha$ are interpreted as creation and annihilation operators with the commutation rules $[p_k^\alpha,p_j^\beta]= i \hbar k \eta^{\alpha\beta} \delta_{k+j,0}$, where $\eta^{\alpha\beta}$ is a symmetric nondegenerate matrix. Its integrability is inherited by the classical limit and consists in an infinite family of commuting hamiltonian operators $\oG_{\beta,d}$, $\beta=1,\ldots,N$, $d=-1,0,1,\ldots$.\\

In the paper, after introducing the relevant deformation quantization of the standard hydrodynamic Poisson bracket of dispersionless DZ hierarchies and a study of its propagator (see \cite{DZ05}), we define the qDR hierarchy using intersection numbers of a given CohFT with the double ramification cycle and the Hodge and psi classes and prove commutativity of the (quantum) flows. We also prove a quantum version of the recursion relations from \cite{BR14} which allow to reconstruct the entire qDR hierarchy starting from $\oG_{1,1}$ alone, the hamiltonian operator associated to the first descendant of the unit in the CohFT.  Thanks also to this recursion we are able to study various examples, effectively computing the quantization of the KdV, Intermediate Long Wave and extended Toda hierarchies. For the dispersionless limit of the KdV hierarchy we are able to prove that an explicit generating function for the quantum Hamiltonians (which appeared in \cite{BSSZ12}, see aso \cite{Ros07}) satisfies indeed our recursion. Finally, in an Appendix, we prove that the integral of an arbitrary tautological class over the double ramification cycle $DR_g(a_1,\ldots,a_n)$ is a polynomial in the ramification multiplicities $a_1,\ldots,a_n$. This means, for instance, that for tautological CohFTs the qDR Hamiltonian densities are given by differential polynomials, a general assumption under which we work in this paper.

\vspace{0.5cm}

\noindent{\bf Acknowledgments.}\\
We would like to thank Boris Dubrovin, Rahul Pandharipande, Vladimir Rubtsov, Sergey Shadrin and Dimitri Zvonkine for useful discussions. P. R. was partially supported by a Chaire CNRS/Enseignement superieur 2012-2017 grant. A. B. was supported by grant ERC-2012-AdG-320368-MCSK in the group of R.~Pandharipande at ETH Zurich, by the Russian Federation Government grant no. 2010-220-01-077 (ag. no. 11.634.31.0005), the grants RFFI 13-01-00755 and NSh-4850.2012.1.

Part of the work was completed during the visit of A.B to the University of Burgundy in 2014 and during the visit of P.R. to the Forschungsinstitut f\"ur Mathematik at ETH Z\"urich in 2014.

\vspace{0.5cm}

\section{Deformation quantization of the standard hydrodynamic Poisson bracket}

In this section we describe a deformation quantization of a Poisson algebra which appears in multiple contexts. It is the natural Poisson structure arising from the algebriac setting of rational Symplectic Field Theory \cite{EGH00} (see for instance \cite{Ros09}), it is the normal form of any $(0,n)$-Poisson brackets of hydrodynamic type according to the classification of \cite{DZ05} and it is the relevant Poisson structure for the double ramification hierarchy associated to a given CohFT, according to \cite{Bur14}.\\

\subsection{The standard hydrodynamic Poisson bracket}
Referring the reader to \cite{DZ05,Bur14,BR14} for precise definitions, we recall that a function $f=f(u^\alpha,u^\alpha_x,u^\alpha_{xx},\ldots;\eps)$ is a differential polynomial in the jets $u^\alpha_k$, $\alpha=1,\ldots,N$, $k=0,1,2,\ldots$, where $u^\alpha = u^\alpha_0, u^\alpha_x=u^\alpha_1, \ldots$, if $f$ is a formal power series in $\eps$ with coefficients that are polynomials in $u^\alpha_k$, for $k>0$, and power series in $u^\alpha_0$. The degree of a differential polynomial is determined by setting $\deg u^\alpha_i = i$ and $\deg\eps = -1$. The space of local functionals is given by the quotient of the space of differential polynomials first by constants and then by the image of the formal $x$-derivative operator $\partial_x = \sum_{k\geq 0} u^\alpha_{k+1} \frac{\partial}{\partial u^\alpha_k}$ (we adhere, here and in the following, to the convention of sum over repeated greek indices, but not over latin ones). A local functional is usually represented with the symbol of an integral over the circle $\overline{f}=\int f dx$, to represent the fact that the equivalence class of $\partial_x f$ is zero. On the space of local functionals we have the \emph{standard hydrodynamic Poisson bracket} associated with a nondegenerate symmetric matrix~$\eta^{\alpha \beta}$:
$$
\{\overline{f}, \og \} := \int\left( \frac{\delta \overline{f}}{\delta u^\alpha} \eta^{\alpha \beta} \partial_x \frac{\delta \overline{g}}{\delta u^\beta} \right)dx,
$$
where $\frac{\delta \overline{f}}{\delta u^\alpha} := \sum_{k\geq 0} (-\partial_x)^k \frac{\partial f}{\partial u^\alpha_k}$. With respect to this Poisson bracket, the time evolution of a differential polynomial along the flow generated by a local functional $\og$, the Hamiltonian, is given by $\frac{\partial f}{\partial t} = \{f,\og\}:=\sum_{k\geq 0} \frac{\partial f}{\partial u^\alpha_k} \eta^{\alpha \beta} \partial_x^{k+1} \frac{\delta \overline{g}}{\delta u^\beta}$.\\

\subsection{Quantum commutator on local functionals}
In \cite{EGH00} the authors described a Weyl algebra formed by (power series in $\hbar$ with coefficients that are) power series in $p^\alpha_k$, $k\leq 0$, with coefficients that are polynomials in $p^\alpha_k$, $k>0$, with $\alpha=1,\ldots,N$. The product rule is described as follows: representing two power series in the ``normal form'', i.e. with all variables with negative or zero subscripts appearing on the left of all variables with positive subscripts,
$$
f=\sum_{g\geq 0} \sum_{n\geq 0}\ \sum_{k_1,\ldots, k_n \leq 0} p^{\alpha_1}_{k_1}\ldots p^{\alpha_n}_{k_n} f^{\alpha_1,\ldots,\alpha_n}_{k_1,\ldots,k_n;g}(p_{k>0})\hbar^g,
$$
$$
g=\sum_{g \geq 0} \sum_{n\geq 0}\ \sum_{k_1,\ldots, k_n\leq 0} p^{\alpha_1}_{k_1}\ldots p^{\alpha_n}_{k_n} g^{\alpha_1,\ldots,\alpha_n}_{k_1,\ldots,k_n;g}(p_{k>0}) \hbar^g,
$$
where $f^{\alpha_1,\ldots,\alpha_n}_{k_1,\ldots,k_n;g}(p_{>0})$ and $g^{\alpha_1,\ldots,\alpha_n}_{k_1,\ldots,k_n;g}(p_{>0})$ are polynomials, one obtains the product $f\star g$ by commuting the $p_{\leq 0}$ variables of $g$ with the $p_{k>0}$ variables of $f$ using $[p^\alpha_k,p^\beta_j]=i \hbar k \eta^{\alpha \beta}\delta_{k+j,0} $. Thanks to polynomiality of the coefficients, this process is well defined and produces another element of the same Weyl algebra.\\

This definition of the $\star$-product coincides with endowing 
$$
\C[p^1_{k>0},\ldots,p^N_{k>0}][[p^1_{k\leq 0},\ldots,p^N_{k\leq 0},\hbar]]
$$ with the ``normal ordering'' $\star$-product
$$
f \star g =f \left(  e^{\sum_{k>0} i \hbar k \eta^{\alpha \beta} \overleftarrow{\frac{\partial }{\partial p^\alpha_{k}}} \overrightarrow{\frac{\partial }{\partial p^\beta_{-k}}}}\right) g.
$$
The commutator is then defined consequently as $[f,g]:=f\star g - g \star f$.\\

We now want to describe how this $\star$-product is translated to the language of differential polynomials and local functionals. For the relevant definitions and notations we refer to \cite{DZ05}, but also to our previous papers \cite{Bur14,BR14}. We will need, first, to extend the space of differential polynomials to allow for dependence on the quantization formal parameter $\hbar$. In view of the results of Appendix \ref{subsection:polynomiality} we will make the following choice: a \emph{quantum differential polynomial} $f=f(u^\alpha,u^\alpha_x,\ldots;\eps,\hbar)$ is a formal power series in $\hbar$ and $\epsilon$ whose coefficients are polynomials in $u^\alpha_k$, for $k>0$, and power series in $u^\alpha_0$. The quantization parameter has degree $\deg\hbar=-2$. The space of \emph{quantum local functionals} is given, as in the classical case, by taking quotients with respect to constants and the image of the $\partial_x$-operator. Consider now a change of variables
$$
u^\alpha_j=\sum_{k\in\mbZ}(ik)^jp^\alpha_k e^{ikx},
$$
which allows to express any quantum differential polynomial $f=f(u^\alpha,u^\alpha_x,\ldots;\eps,\hbar)$ as a formal Fourier series in~$x$ with coefficients that are (power series in~$\eps$ with coefficients) in the above Weyl algebra. When needed, we will stress the dependence of~$f$ on the formal variable~$x$ by writing~$f(x)$. Note that
$$
\frac{\d}{\d p^\alpha_k}f(x)=\sum_{s\ge 0}(ik)^s e^{ikx}\frac{\d f}{\d u^\alpha_s}(x).
$$
Therefore, for any two differential polynomials $f$ and $g$ we have
\begin{equation*}
f(x)\star g(y) =\sum_{\substack{n\geq 0\\ r_1,\ldots,r_n\geq 0\\ s_1,\ldots , s_n\geq 0}} \frac{\hbar^{n}}{n!} \frac{\partial^n f}{\partial u^{\alpha_1}_{s_1}\ldots \partial u^{\alpha_n}_{s_n}}(x)\left( \prod_{k=1}^n (-1)^{r_k}  \eta^{\alpha_k\beta_k} \delta_+^{(r_k + s_k +1)}(x-y) \right)  \frac{\partial^n g}{\partial u^{\beta_1}_{r_1}\ldots \partial u^{\beta_n}_{r_n}}(y), 
\end{equation*}
where
$$
\delta_+^{(s)}(x-y):= \sum_{k\geq 0} (ik)^s e^{i k (x-y)}, \hspace{1cm} s\geq 0.
$$
The product of derivatives of the formal Fourier series $\delta_+(x-y)$ can be expressed as a linear combination of derivatives of the same object:
$$
\delta_+^{(a_1)}(x-y)\ \ldots \ \delta_+^{(a_n)}(x-y) =\sum_{j=1}^{n-1+\sum_{k=1}^n a_k}(-i)^{n-1} C_j^{a_1,\ldots,a_n}\ \delta^{(j)}_+(x-y),
$$
where $a_1,\ldots,a_n\geq 1$ and the $C_j^{a_1,\ldots,a_k}$ are rational numbers such that $C_j^{a_1,\ldots,a_n}=0$ unless $j=n-1+\sum_{k=1}^n a_k\ (\mathrm{mod}\ 2)$ (see Appendix~\ref{subsection:products} and Lemma~\ref{lemma:products of derivatives} there). For any specific choice of $a_1,\ldots,a_n$ it is not hard to compute the coefficients~$C^{a_1,\ldots,a_n}_j$. We also have a general formula for the top coefficient
$$
C^{a_1,\ldots,a_n}_{n-1+\sum_{k=1}^n a_k}= \frac{\prod_{k=1}^n a_k!}{(n-1+\sum_{k=1}^n a_k)!}
$$
Since the parities of the number of $x$-derivatives in the linear combination all agree, using the fact that
$$
\delta^{(s)}(x-y) = \delta_+^{(s)}(x-y) + \delta_-^{(s)}(x-y),
$$
where $\delta^{(s)}(x-y):=\sum_{k\in \Z} (ik)^s e^{ik(x-y)}$ is the $s$-th derivative of the formal periodic Dirac delta-function, $\delta^{(s)}_- (x-y):=\sum_{k< 0} (ik)^s e^{ik(x-y)}$ and $ \delta_+^{(s)}(x-y) = (-1)^s \delta^{(s)}_- (y-x)$, for $s>0$, we obtain 
\begin{equation}
\begin{split}
[f(x),g(y)]=\sum_{\substack{n\geq 1\\ r_1,\ldots ,r_n\geq 0\\ s_1,\ldots,s_n\geq 0}} \frac{(-i)^{n-1} \hbar^{n}}{n!}  &\frac{\partial^n f}{\partial u^{\alpha_1}_{s_1}\ldots \partial u^{\alpha_n}_{s_n}}(x)  (-1)^{\sum_{k=1}^n r_k}  \left( \prod_{k=1}^n \eta^{\alpha_k \beta_k}\right) \times\\
& \times \sum_{j=1}^{2n-1+\sum_{k=1}^n (s_k+r_k)} C_j^{s_1+r_1+1,\ldots,s_n+r_n+1}\delta^{(j)}(x-y)\frac{\partial^n g}{\partial u^{\beta_1}_{r_1}\ldots \partial u^{\beta_n}_{r_n}}(y).
\end{split}
\end{equation}
In particular, we get
\begin{equation}
\begin{split}
[f,\og]=\sum_{\substack{n\geq 1\\ r_1,\ldots ,r_n\geq 0\\ s_1,\ldots,s_n\geq 0}} \frac{(-i)^{n-1} \hbar^{n}}{n!}  &\frac{\partial^n f}{\partial u^{\alpha_1}_{s_1}\ldots \partial u^{\alpha_n}_{s_n}}  (-1)^{\sum_{k=1}^n r_k}  \left( \prod_{k=1}^n \eta^{\alpha_k \beta_k}\right) \times\\
& \times \sum_{j=1}^{2n-1+\sum_{k=1}^n (s_k+r_k)} C_j^{s_1+r_1+1,\ldots,s_n+r_n+1}  \partial_x^j  \frac{\partial^n g}{\partial u^{\beta_1}_{r_1}\ldots \partial u^{\beta_n}_{r_n}}.
\end{split}
\end{equation}
Notice how, for any quantum differential polynomial $f$ and quantum local functional $\og$, this formula gives a quantum differential polynomial of degree $\deg( [f,\og]) = \deg f + \deg \og - 1$.\\

Taking the classical limit of this expression one obtains $\left(\frac{1}{\hbar}[\overline{f},\og]\right)|_{\hbar=0}=\{\overline{f}|_{\hbar=0},\og|_{\hbar=0}\}$, i.e. the standard hydrodynamic Poisson bracket on the classical limit of the local functionals.\\

\section{Quantum double ramification hierarchy}
\subsection{Hamiltonian densities}
Given a cohomological field theory $c_{g,n}\colon V^{\otimes n} \to H^{even}(\oM_{g,n};\mbC)$, we define the hamiltonian densities of the \emph{quantum double ramification hierarchy} as the following generating series:
\begin{equation}\label{density}
\begin{split}
G_{\alpha,d}:=&\sum_{\substack{g\ge 0,n\ge 0\\2g-1+n>0}}\frac{(i \hbar)^g}{n!}\times\\
&\times\sum_{\substack{a_1,\ldots,a_n\in\mbZ\\ \alpha_1,\ldots,\alpha_n}}\left(\int_{DR_g\left(-\sum a_i,a_1,\ldots,a_n\right)}\Lambda\left(\frac{-\eps^2}{i \hbar}\right) \psi_1^d c_{g,n+1}\left(e_\alpha\otimes\otimes_{i=1}^n e_{\alpha_i}\right)\right)p^{\alpha_1}_{a_1}\ldots p^{\alpha_n}_{a_n}e^{ix\sum a_i},
\end{split}
\end{equation}
for $\alpha=1,\ldots,N$ and $d=0,1,2,\ldots$. Here $DR_g\left(a_1,\ldots,a_n\right) \in H^{2g}(\oM_{g,n};\Q)$ is the double ramification cycle, $\Lambda\left(\frac{-\eps^2}{i \hbar}\right):=\left(1+ \left( \frac{-\eps^2}{i \hbar}\right) \lambda_1+\ldots + \left(\frac{-\epsilon^2}{i\hbar}\right)^g \lambda_g \right)$, with $\lambda_i$ the $i$-th Chern class of the Hodge bundle and $\psi_i$ is the first Chern class of the tautological bundle at the $i$-th marked point. When needed, we will stress the dependence of $G_{\alpha,d}$ on the formal variable $x$ by writing $G_{\alpha,d}(x)$.\\

Our definition of quantum double ramification hierarchy might probably also be referred to as symplectic field theory hierarchy associated to the cohomological field theory 
$$
c_{g,n}\left(\otimes_{i=1}^n e_{\alpha_i}\right)\times\left(1+ \left( \frac{-\eps^2}{i \hbar}\right) \lambda_1+\ldots + \left(\frac{-\eps^2}{i\hbar}\right)^g \lambda_g \right),
$$ given that, when we take the integral with respect to $x$ of the above generating series and we take as cohomological field theory the Gromov-Witten theory of a closed target symplectic manifold $M$, we obtain indeed the definition of the SFT Hamiltonians \cite{EGH00} for the standard stable hamiltonian structure on $M\times S^1$ (see also \cite{FR10}).\\

As for the ``classical'' hamiltonian densities $g_{\alpha,p}=G_{\alpha,p}|_{\hbar=0}$ defined in \cite{BR14}, we would like to rewrite the above expression in terms of formal jet variables $u^\alpha_s = \sum_{k\in\mbZ} (ik)^s p^\alpha_k e^{ikx}$, $\alpha=1,\ldots,N$, $s=0,1,2,\ldots$. Working under the assumption that the double ramification cycle $DR_{g}(a_1,\ldots,a_n)$ is a non-homogeneous polynomial of degree at most~$2g$ in the variables $a_1,\ldots,a_n$ (which is compatible with the recent conjecture by Pixton on the double ramification cycle explicit form), we actually obtain that each $G_{\alpha,p}$ can be uniquely written as a quantum differential polynomial of degree $\deg G_{\alpha,p} \leq 0$. This means that the number of $x$-derivatives that can appear in the coefficient of $\eps^k\hbar^j$ is at most ~$k+2j$. In fact, it is proved  in Appendix \ref{subsection:polynomiality} that  integrals of tautological classes over the double ramification cycle $DR_{g}(a_1,\ldots,a_n)$ are indeed non-homogeneous polynomials of degree at most $2g$ in the variables $a_1,\ldots,a_n$. So, if the cohomological field theory we start with is tautological, this ensures that our densities are quantum differential polynomials of non-positive degree.

\begin{remark}
Cohomological field theories that consist of tautological classes form a very large class of cohomological field theories. In particular, all semisimple cohomological field theories and also cohomological field theories whose shift is semisimple belong to this class (see e.g.~\cite{PPZ15}).
\end{remark}

We finally add manually $N$ extra densities $G_{\alpha,-1}:=\eta_{\alpha\mu} u^\mu$. Recall that by $\oG_{\alpha,p}= \int G_{\alpha,p} dx$ we denote the coefficient of $e^{i0x}$ in $G_{\alpha,p}$ considered also up to a constant, for all $\alpha=1,\ldots,N$, $p=-1,0,1,\ldots$.

\subsection{Main Lemma}
In order to prove commutativity and recursion formulae for the quantum double ramification hierarchy one can proceed exactly as in \cite{Bur14} and \cite{BR14}, respectively. However here we will take a slightly different approach based on the following result from \cite{BSSZ12}. For a subset $I=\{i_1,i_2,\ldots\}, i_1<i_2<\ldots$, of the set $\{1,2,\ldots,n\}$ we will use the following notations:
$$
A_I:=(a_{i_1},a_{i_2},\ldots),\qquad a_I:=\sum_{i\in I}a_i.
$$
Suppose the set~$\{1,2,\ldots,n\}$ is divided into two disjoint subsets, $I\sqcup J=\{1,2,\ldots,n\}$, in such a way that $a_I>0$. Choose a list of positive integers~$k_1,\ldots,k_p$ such that $\sum_{i=1}^p k_i=a_I$. Let us denote by $DR_{g_1}(A_I,-k_1, \dots, -k_p)\boxtimes DR_{g_2}(A_J, k_1, \dots, k_p)$ the cycle in~$\oM_{g_1+g_2+p-1,n}$ obtained by gluing the two double ramification cycles at the marked points labeled by~$k_1,\ldots,k_p$.

\begin{theorem}[\cite{BSSZ12}] \label{BSSZ}
Let $t$ and $s$ be two different elements in $\{ 1, \dots, n \}$. Assume that both $a_s$ and $a_t$ are non-zero. Then we have
\begin{align*}
(a_s \psi_s - a_t \psi_t) DR_g&(a_1, \dots, a_n)=\\
=&\sum_{s\in I, t\in J}\sum_{p\geq 1}\sum_{g_1, g_2}\sum_{k_1,\dots,k_p}
\frac{\prod_{i=1}^p k_i}{p!}
DR_{g_1}(A_I,-k_1, \dots, -k_p)\boxtimes DR_{g_2}(A_J, k_1, \dots, k_p)\\
&-
\sum_{t \in I, s\in J}
\sum_{p \geq 1}\sum_{g_1, g_2} 
\sum_{k_1, \dots, k_p}
\frac{\prod_{i=1}^p k_i}{p!}
DR_{g_1}(A_I,-k_1, \dots, -k_p) \boxtimes DR_{g_2}(A_J, k_1, \dots, k_p).
\end{align*}
where the first sum is taken over all $I \sqcup J = \{ 1, \dots, n\}$ such that $a_I >0$; the third sum is over all non-negative genera $g_1$, $g_2$ satisfying $g_1 + g_2 + p-1= g$; the fourth sum is over the $p$-uplets of positive integers with total sum $a_I=-a_J$.
\end{theorem}

We now define the following generating function for intersection numbers involving the insertion of psi-classes at two marked points:
\begin{equation*}\label{density2}
\begin{split}
G_{\alpha,p;\beta,q}(x,y):=\sum_{\substack{g\ge 0,n\ge 0\\2g+n>0}}\frac{(i\hbar)^g}{n!} \sum_{\substack{a_0,\ldots,a_{n+1}\in\mbZ\\\sum a_i=0\\ \alpha_1,\ldots,\alpha_n}}&\left(\int_{DR_g\left(a_0,a_1,\ldots,a_n,a_{n+1}\right)}\right. \Lambda\left(\frac{-\eps^2}{i \hbar}\right) \psi_0^p \psi_{n+1}^q \times \\
&\times c_{g,n+2}\left(e_\alpha\otimes\otimes_{i=1}^n e_{\alpha_i}\otimes e_\beta\right)\Bigg) p^{\alpha_1}_{a_1}\ldots p^{\alpha_n}_{a_n}e^{-i a_0 x-i a_{n+1} y},
\end{split}
\end{equation*}
for $\alpha,\beta=1,\ldots,N$ and $p,q=0,1,2,\ldots$. Then we have the following
\begin{lemma}\label{mainlemma}
For all  $\alpha,\beta=1,\ldots,N$ and $p,q=0,1,2,\ldots$, we have
\begin{equation}\label{eqmainlemma}
\partial_x G_{\alpha,p+1;\beta,q}(x,y) - \partial_y G_{\alpha,p;\beta,q+1}(x,y) =\frac{1}{\hbar} \left[ G_{\alpha,p}(x) , G_{\beta,q}(y)\right]
\end{equation}
\end{lemma}
\begin{proof}
The proof is a simple consequence of the definition of the generating series and the application of Theorem \ref{BSSZ}.
\end{proof}

\subsection{Commutativity and recursion}

As consequences of Lemma \ref{mainlemma} we find
\begin{theorem}\label{commutativity}
For all  $\alpha,\beta=1,\ldots,N$ and $p,q=-1,0,1,\ldots$, we have
$$\left[ \oG_{\alpha,p},\oG_{\beta,q} \right] = 0$$
\end{theorem}
\begin{proof}
One simply integrates equation (\ref{eqmainlemma}) with respect to both $x$ and $y$. Notice that $\oG_{\alpha,-1}$ is a Casimir of the standard quantum commutator and hence commutes automatically with all other~$\oG_{\beta,q}$.
\end{proof}

\begin{theorem}\label{theorem:recursion}
For all $\alpha=1,\ldots,N$ and $p=-1,0,1,\ldots$, we have
\begin{gather}\label{eq:first recursion}
\partial_x (D-1) G_{\alpha,p+1} =\frac{1}{\hbar} \left[ G_{\alpha,p} , \oG_{1,1} \right],
\end{gather}
\begin{equation}\label{eq:second recursion}
\partial_x \frac{\partial G_{\alpha,p+1}}{\partial u^\beta} =\frac{1}{\hbar} \left[G_{\alpha,p}, \oG_{\beta,0} \right],
\end{equation}
where $D:=\eps\frac{\partial}{\partial\eps} + 2\hbar\frac{\partial}{\partial \hbar} + \sum_{s\ge 0} u^\alpha_s\frac{\partial}{\partial u^\alpha_s}$.
\end{theorem}
\begin{proof}
Suppose $p\ge 0$. For both formulae one needs to integrate with respect to $y$ equation (\ref{eqmainlemma}) with $(\beta,q)=(1,1)$ and $(\beta,q)=(\beta,0)$ respectively and, for the first equation, use the following version of the divisor equation 
\begin{gather}\label{eq:dil}
\int G_{\alpha,p+1;1,1}(x,y) dy = (D-1) G_{\alpha,p+1}(x),
\end{gather}
while for the second one, the fact that, by definition, 
\begin{gather}\label{eq:der}
\int G_{\alpha,p+1;\beta,0}(x,y) dy= \frac{\partial G_{\alpha,p+1}(x)}{\partial u^\beta}.
\end{gather}
If $p=-1$, then we have
$$
\frac{1}{\hbar}[G_{\alpha,-1},\oG_{\beta,q}]=\frac{1}{\hbar}\left[\sum_k \eta_{\alpha\mu}p^\mu_k e^{ikx},\oG_{\beta,q}\right]=\sum_k ik\frac{\d\oG_{\beta,q}}{\d p^\alpha_{-k}}e^{ikx}=\int\d_x G_{\alpha,0;\beta,q}(x,y)dy.
$$ 
We again finish the proof using equation~\eqref{eq:dil} or equation~\eqref{eq:der}. 
\end{proof}

\subsection{String equation}

We have the following immediate generalization of a Lemma from \cite{Bur14}.
\begin{lemma}
We have $\oG_{1,0} = \frac{1}{2}\int \left(\eta_{\mu \nu} u^\mu u^\nu \right) dx$, so that, for any quantum differential polynomial $f=f(u^\alpha,u^\alpha_x,\ldots)$, we have $\frac{1}{\hbar}\left[f, \oG_{1,0} \right] = \partial_x f$.
\end{lemma}
\begin{proof}
The proof is the same as for the classical limit. Since $DR_g\left(0,a_1,\ldots,a_n\right) = \pi^* DR_g\left(a_1,\ldots,a_n\right)$, $c_{g,n+1}\left(e_1\otimes\otimes_{i=1}^n e_{\alpha_i}\right) = \pi^* c_{g,n}\left(\otimes_{i=1}^n e_{\alpha_i}\right)$ and $\Lambda(\epsilon)=\pi^* \Lambda(\epsilon)$, where $\pi:\oM_{g,n+1}\to\oM_{g,n}$, the contribution to $\oG_{1,0}$ vanishes for $g>0$ or $n>2$. For $g=0$ and $n=2$ the result follows from $c_{0,3}(e_1,e_\mu,e_\nu)=\eta_{\mu \nu}$.
\end{proof}

We also have the following version of the string equation for the quantum double ramification hierarchy.
\begin{lemma}
For all $\alpha=1,\ldots,N$ and $d=-1,0,1,\ldots$ we have
\begin{equation}\label{eq:string}
G_{\alpha,d} = \frac{\partial G_{\alpha,d+1}}{\partial u^1}.
\end{equation}
\end{lemma}
\begin{proof}
We have
\begin{gather*}
\begin{split}
\frac{\d G_{\alpha,d+1}}{\d u^1}=&\sum_{\substack{g,n\ge 0\\2g+n>0}}\frac{(i\hbar)^g}{n!}\times\\
&\sum_{\substack{a_0,a_1,\ldots,a_n\in\mbZ\\ \alpha_1,\ldots,\alpha_n}}\left(\int_{DR_g\left(a_0,a_1,\ldots,a_n,0\right)}\Lambda\left(\frac{-\eps^2}{i \hbar}\right) \psi_0^{d+1} c_{g,n+2}\left(e_\alpha\otimes\otimes_{i=1}^n e_{\alpha_i}\otimes e_1\right)\right)\left(\prod_{i=1}^n p^{\alpha_i}_{a_i}\right)e^{-ia_0x}.
\end{split}
\end{gather*}
The same argument as in the proof of the previous lemma shows that
\begin{multline*}
\int_{DR_g\left(a_0,a_1,\ldots,a_n,0\right)}\Lambda\left(\frac{-\eps^2}{i \hbar}\right) \psi_0^{d+1} c_{g,n+2}\left(e_\alpha\otimes\otimes_{i=1}^n e_{\alpha_i}\otimes e_1\right)=\\
=
\begin{cases}
\int_{DR_g\left(a_0,a_1,\ldots,a_n\right)}\Lambda\left(\frac{-\eps^2}{i \hbar}\right) \psi_0^d c_{g,n+1}\left(e_\alpha\otimes\otimes_{i=1}^n e_{\alpha_i}\right),&\text{if $2g+n>1$ and $d\ge 0$},\\
0,&\text{if $2g+n>1$ and $d=-1$},\\
\delta_{d,-1}\eta_{\alpha\alpha_1},&\text{if $g=0$ and $n=1$}.
\end{cases}
\end{multline*}
The lemma is proved.
\end{proof}

\subsection{Reconstruction of the hierarchy from $\oG_{1,1}$}\label{subsection:reconstruction}

Notice that equation (\ref{eq:first recursion}) allows us to recover the hamiltonian density $G_{\alpha,d+1}$ up to a constant starting from the knowledge of $G_{\alpha,d}$ and $\oG_{1,1}$. On the other hand recursion~\eqref{eq:first recursion} is insensitive to the constant term in $G_{\alpha,d}$. We see that the all hamiltonian densities $G_{\alpha,d}$ can be determined up to constants starting from the knowledge of $\oG_{1,1}$ alone (and the fact that $G_{\alpha,-1}=\eta_{\alpha\mu} u^\mu$ for any CohFT). The constant terms can be recovered using the string equation~\eqref{eq:string}. We conclude that, once we compute $\oG_{1,1}$, the quantum double ramification hierarchy is completely identified.

%%%%%%%%%%%%%%%%%%%%%%%%%%%%%%%%%%%%%%%%%%%%%%%%%%%%%%%%%%%%%%%%%%%%%
%%%%%%%%%%%%%%%%%%%%%%%%%%%%%%%%%%%%%%%%%%%%%%%%%%%%%%%%%%%%%%%%%%%%%

\section{Examples}

In this section we consider several examples of the quantum double ramification hierarchies. In all these examples the corresponding cohomological field theory is semisimple. This ensures that the densities of the quantum double ramification hierarchies in these cases are quantum differential polynomials.

\subsection{Quantum KdV hierarchy}

Consider the simplest cohomological field theory: 
$$
V=\<e_1\>,\quad {\eta_{1,1}=1},\quad c_{g,n}(e_1^{\otimes n})=1.
$$
In~\cite{Bur14} it was proved that the corresponding double ramification hierarchy coincides with the Korteweg-de Vries hierarchy. Therefore, our quantum double ramification hierarchy in this case gives a quantization of the KdV hierarchy. In Section~\ref{subsubsection:full quantum KdV} we compute several first quantum Hamiltonians for the full hierarchy. In Section~\ref{subsubsection:dispersionless quantum KdV} we check an explicit formula for all quantum densities in the dispersionless limit.

We will omit the first index in the densities~$G_{\alpha,d}$ and in the Hamiltonians~$\oG_{\alpha,d}$.

\subsubsection{Full hierarchy}\label{subsubsection:full quantum KdV}

Let us compute the Hamiltonian $\oG_1$. For the classical part~$\left.\oG_1\right|_{\hbar=0}$ we have (see~\cite{Bur14}):
$$
\left.\oG_1\right|_{\hbar=0}=\og_1=\int\left(\frac{u^3}{6}+\frac{\eps^2}{24}u u_{xx}\right)dx.
$$
Let us compute the quantum correction. For this we have to compute the integrals
$$
\int_{DR_g(0,a_1,\ldots,a_n)}\psi_1\lambda_j,
$$
for $g,n\ge 1$ and $j<g$. The dimension constraint says that $2g-2+n=1+j$. This equation holds only if $g=1, j=0$ and $n=1$. Therefore,
$$
\oG_1-\og_1=i\hbar\left(\int_{DR_g(0,0)}\psi_1\right)p_0=-i\hbar\left(\int_{\oM_{1,2}}\psi_1\lambda_1\right)p_0=-\frac{i\hbar}{24}\int udx.
$$
We conclude that
$$
\oG_1=\int\left(\frac{u^3}{6}+\frac{\eps^2}{24}u u_{xx}-\frac{i\hbar}{24}u\right)dx.
$$
According to Section~\ref{subsection:reconstruction}, this allows to compute all the densities~$G_d$. The first few are
\begin{align*}
G_0=&\frac{u^2}{2}+\frac{\eps^2}{24}u_{xx}-\frac{i\hbar}{24},\\
G_1=&\frac{u^3}{6}+\frac{\eps^2}{24}u u_{xx}+\frac{\eps^4}{1152}u_{xxxx}-i\hbar\frac{u+u_{xx}}{24}-\frac{i\hbar\eps^2}{2880},\\
G_2=&\frac{u^4}{24}+\eps^2\frac{u^2 u_2}{48}+\eps^4\left(\frac{7u_2^2}{5760}+\frac{u u_4}{1152}\right)+\eps^6 \frac{u_6}{82944}-i\hbar\frac{2uu_2+u^2}{48}-i\hbar\eps^2\frac{u+5u_2+4u_4}{2880}\\
&-\frac{i\hbar\eps^4}{120960}+(i\hbar)^2\frac{7}{5760}.
\end{align*}
In particular, for the quantum Hamiltonian $\oG_2$ we get
$$
\oG_2=\int\left(\frac{u^4}{24}+\eps^2\frac{u^2 u_2}{48}+\eps^4\frac{u u_4}{480}-i\hbar\frac{2uu_2+u^2}{48}-i\hbar\eps^2\frac{u}{2880}\right)dx.
$$

\subsubsection{Dispersionless hierarchy}\label{subsubsection:dispersionless quantum KdV}

Consider the dispersionless part of the quantum KdV hierarchy. Let $G_d^{[0]}:=\left.G_d\right|_{\eps=0}$. Introduce the generating series
$$
G^{[0]}(y)=\sum_{d\ge -1}y^d G_d^{[0]}.
$$
Let $S(z):=\frac{e^{\frac{z}{2}}-e^{-\frac{z}{2}}}{z}$.
\begin{proposition}
We have
\begin{gather}\label{eq:generating series}
G^{[0]}(y)=\frac{1}{y^2S(\sqrt{i}\lambda y)}e^{yS\left(\frac{\lambda}{\sqrt{i}}y\d_x\right)u}-y^{-2},
\end{gather}
where $\lambda$ is a formal variable such that $\lambda^2=\hbar$.
\end{proposition}
\begin{proof}
The proposition is an immediate consequence of the following formula that was proved in~\cite{BSSZ12}. Suppose $g,n\ge 0$ and $2g-1+n>0$, then we have
$$
\int_{DR_g(-\sum_{i=1}^n a_i,a_1,\ldots,a_n)}\psi_1^{2g-2+n}=\Coef_{z^{2g}}\left(\frac{\prod_{i=1}^n S(a_iz)}{S(z)}\right).
$$
However, we would like to present another proof of equation~\eqref{eq:generating series} that is based on our recursion from Theorem~\ref{theorem:recursion}. We can reformulate recursion~\eqref{eq:first recursion} in the following way:
\begin{gather}\label{eq:identity}
\d_x\frac{\d}{\d y}(y G^{[0]}(y))=\frac{y}{\hbar}\left[G^{[0]}(y),\overline G^{[0]}_1\right].
\end{gather}
Let us check that the right-hand side of~\eqref{eq:generating series} satisfies this equation.

Denote the coefficients of the power series~$S(z)$ by $s_i$: $S(z)=\sum_{i\ge 0}s_i z^{2i}$. Since $\oG^{[0]}_1=\int\left(\frac{u^3}{6}-\frac{i\hbar}{24}u\right)dx$, if we substitute~\eqref{eq:generating series} on the right-hand side of~\eqref{eq:identity}, we get
\begin{multline}\label{eq:tmp1}
\frac{1}{S(\sqrt i\lambda y)}\sum_{r\ge 0}\left(\frac{\lambda^2 y^2}{i}\right)^{r}s_r\d_x^{2r+1}\left(\frac{u^2}{2}\right)e^{yS\left(\frac{\lambda}{\sqrt i} y\d_x\right)u}+\\
+\frac{y\lambda^2}{2S(\sqrt i\lambda y)}\sum_{\substack{r_1,r_2\ge 0\\j\ge 1}}i^{-1-r_1-r_2}(\lambda y)^{2r_1+2r_2}s_{r_1}s_{r_2}C^{2r_1+1,2r_2+1}_j u_je^{yS\left(\frac{\lambda}{\sqrt i} y\d_x\right)u}.
\end{multline}
In order to shorten computations a little bit, let us make the following rescalings:
$$
\d_x\mapsto i\d_x,\qquad \lambda\mapsto\frac{\lambda}{\sqrt{i}}.
$$
Therefore, we have to prove that
\begin{align}
\d_x\frac{\d}{\d y}\left(\frac{1}{y S(\lambda y)}e^{yS(\lambda y\d_x)u}\right)=&\frac{1}{S(\lambda y)}\sum_{r\ge 0}(\lambda y)^{2r}s_r\d_x^{2r+1}\left(\frac{u^2}{2}\right)e^{yS(\lambda y\d_x)u}+\label{eq:identity2}\\
&+\frac{y\lambda^2}{2S(\lambda y)}\sum_{\substack{r_1,r_2\ge 0\\j\ge 1}}(\lambda y)^{2r_1}(\lambda y)^{2r_2}s_{r_1}s_{r_2}\tC^{2r_1+1,2r_2+1}_j u_je^{yS(\lambda y\d_x)u}.\notag
\end{align}
Here we use coefficients $\tC^{a_1,a_2}_j$ introduced in Section~\ref{subsection:sums of powers}. They are related to coefficients $C^{a_1,a_2}_j$ by formula~\eqref{eq:relation of coefficients}. By definition of the function $\tC^{2r_1+1,2r_2+1}(N)$ (see Section~\ref{subsection:sums of powers}), we have
\begin{align*}
\sum_{\substack{r_1,r_2\ge 0\\j\ge 1}}s_{r_1}s_{r_2}(\lambda y)^{2r_1}(\lambda y)^{2r_2}\tC^{2r_1+1,2r_2+1}_j k^j=&\sum_{k_1+k_2=k}S(\lambda y k_1)S(\lambda y k_2)k_1k_2=\\
=&\frac{1}{\lambda^2 y^2}\sum_{k_1+k_2=k}\left(e^{\frac{\lambda yk_1}{2}}-e^{-\frac{\lambda yk_1}{2}}\right)\left(e^{\frac{\lambda yk_2}{2}}-e^{-\frac{\lambda yk_2}{2}}\right)=\\
=&\frac{1}{\lambda^2 y^2}\left((k+1)\left(e^{\frac{\lambda y k}{2}}+e^{-\frac{\lambda y k}{2}}\right)-2\frac{e^{\frac{\lambda y(k+1)}{2}}-e^{-\frac{\lambda y(k+1)}{2}}}{e^{\frac{\lambda y}{2}}-e^{-\frac{\lambda y}{2}}}\right)=\\
=&\frac{1}{\lambda^2 y^2}\left(k\left(e^{\frac{\lambda yk}{2}}+e^{-\frac{\lambda yk}{2}}\right)-\frac{e^{\frac{\lambda y}{2}}+e^{-\frac{\lambda y}{2}}}{e^{\frac{\lambda y}{2}}-e^{-\frac{\lambda y}{2}}}\left(e^{\frac{\lambda yk}{2}}-e^{-\frac{\lambda yk}{2}}\right)\right).
\end{align*}
From this formula it follows that the right-hand side of~\eqref{eq:identity2} is equal to
\begin{align}
&\frac{1}{S(\lambda y)}\d_x S(\lambda y\d_x)\left(\frac{u^2}{2}\right)\cdot e^{yS(\lambda y\d_x)u}+\label{first1}\\
&+\frac{1}{2yS(\lambda y)}\left(\d_x\left(e^{\frac{\lambda y\d_x}{2}}+e^{-\frac{\lambda y\d_x}{2}}\right)-\frac{e^{\frac{\lambda y}{2}}+e^{-\frac{\lambda y}{2}}}{e^{\frac{\lambda y}{2}}-e^{-\frac{\lambda y}{2}}}\left(e^{\frac{\lambda y\d_x}{2}}-e^{-\frac{\lambda y\d_x}{2}}\right)\right)u\cdot e^{yS(\lambda y\d_x)u}.\label{second1}
\end{align}
Let us compute the left-hand side of~\eqref{eq:identity2}. We have
$$
\frac{\d}{\d y}\left(\frac{1}{y S(\lambda y)}e^{yS(\lambda y\d_x)u}\right)=\left[-\frac{\lambda^2}{2}\frac{e^{\frac{\lambda y}{2}}+e^{-\frac{\lambda y}{2}}}{\left(e^{\frac{\lambda y}{2}}-e^{-\frac{\lambda y}{2}}\right)^2}+\frac{\lambda}{2}\frac{e^{\frac{\lambda y\d_x}{2}}+e^{-\frac{\lambda y\d_x}{2}}}{e^{\frac{\lambda y}{2}}-e^{-\frac{\lambda y}{2}}}u\right]e^{yS(\lambda y\d_x)u}.
$$
Therefore,
\begin{align}
\d_x\frac{\d}{\d y}\left(\frac{1}{y S(\lambda y)}e^{yS(\lambda y\d_x)u}\right)=&\frac{\left(e^{\frac{\lambda y\d_x}{2}}+e^{-\frac{\lambda y\d_x}{2}}\right)u\cdot\left(e^{\frac{\lambda y\d_x}{2}}-e^{-\frac{\lambda y\d_x}{2}}\right)u}{2(e^{\frac{\lambda y}{2}}-e^{-\frac{\lambda y}{2}})}e^{yS(\lambda y\d_x)u}+\label{first2}\\
&+\left(\frac{\lambda}{2}\d_x\frac{e^{\frac{\lambda y\d_x}{2}}+e^{-\frac{\lambda y\d_x}{2}}}{e^{\frac{\lambda y}{2}}-e^{-\frac{\lambda y}{2}}}-\frac{\lambda}{2}\frac{e^{\frac{\lambda y}{2}}+e^{-\frac{\lambda y}{2}}}{\left(e^{\frac{\lambda y}{2}}-e^{-\frac{\lambda y}{2}}\right)^2}\left(e^{\frac{\lambda y\d_x}{2}}-e^{-\frac{\lambda y\d_x}{2}}\right)\right)u\cdot e^{yS(\lambda y\d_x)u}.\notag
\end{align}
It is easy to see that~\eqref{first1} is equal to the first summand on the right-hand side of~\eqref{first2} and~\eqref{second1} is equal to the second summand. Therefore, equation~\eqref{eq:identity2} is proved.

By Section~\ref{subsection:reconstruction}, it remains to check that the right-hand side of~\eqref{eq:generating series} satisfies the string equation
$$
\frac{\d G^{[0]}(y)}{\d u}=y G^{[0]}(y)+y^{-1}.
$$
This is a trivial computation. The proposition is proved.
\end{proof}

%%%%%%%%%%%%%%%%%%%%%%%%%%%%%%%%%%%%%%%%%%%%%%%%%%%%%%%%%%%%%%%%%%%%%%%%%%%%%%%%%%%

\subsection{Quantum ILW hierarchy}

Consider the cohomological field theory formed by linear Hodge integrals:
\begin{gather*}
V=\<e_1\>,\quad \eta_{1,1}=1,\quad c_{g,n}\left(e_1^{\otimes n}\right)=1+\mu\lambda_1+\ldots+\mu^g\lambda_g,
\end{gather*}
where $\mu$ is a formal parameter. In~\cite{Bur14,Bur13} it was proved that the corresponding double ramification hierarchy coincides, up to simple rescalings, with the hierarchy of the Intermediate Long Wave equation. In particular,
\begin{gather}\label{eq:classical ILW}
\og_1=\int\left(\frac{u^3}{6}+\sum_{g\ge 1}\eps^{2g}\mu^{g-1}\frac{|B_{2g}|}{2(2g)!}uu_{2g}\right)dx.
\end{gather}
Let us compute the Hamiltonian $\oG_1$ of the quantum double ramification hierarchy.
\begin{lemma}
We have
$$
\oG_1=\int\left(\frac{u^3}{6}+\sum_{g\ge 1}\eps^{2g}\mu^{g-1}\frac{|B_{2g}|}{2(2g)!}uu_{2g}-\frac{i\hbar}{24}u-i\hbar\sum_{g\ge 1}\eps^{2g-2}\mu^g\frac{|B_{2g}|}{2(2g)!}uu_{2g}\right)dx.
$$
\end{lemma}
\begin{proof}
Let us compute the quantum correction to the classical part~\eqref{eq:classical ILW}. For this we have to compute the integrals
$$
\int_{DR_g(0,a_1,\ldots,a_n)}\psi_1\lambda_k\lambda_{g-j},
$$
where $g,j,n\ge 1$ and $k\le g$. The dimension constraint implies that
$$
g=k-j-n+3.
$$
Therefore, we can only have the following possibilities:
\begin{enumerate}
\item[1.] $j=n=1$ and $k=g-1$.
\item[2.] $j=1$, $n=2$ and $k=g$.
\item[3.] $j=2$, $n=1$ and $k=g$.
\end{enumerate}

Consider case 1. We come to the integral
$$
\int_{DR_g(0,0)}\psi_1\lambda_{g-1}^2.
$$
We have $DR_g(0,0)=(-1)^g\lambda_g$. We also have $\lambda_{g-1}^2=2\lambda_g\lambda_{g-2}$, if $g\ge 2$. Since $\lambda_g^2=0$, we conclude that the last integral vanishes, if $g\ge 2$. If $g=1$, then it is equal to $-\frac{1}{24}$ and the corresponding quantum correction is given by
\begin{gather}\label{eq:correction1}
-\int\frac{i\hbar}{24}u dx.
\end{gather}

Consider case 2. We have (see e.g. \cite{CMW12})
$$
\int_{DR_g(0,a,-a)}\psi_1\lambda_g\lambda_{g-1}=a^{2g}\frac{|B_{2g}|}{(2g)!}.
$$
The corresponding quantum correction is equal to
\begin{gather}\label{eq:correction2}
-i\hbar\int\left(\sum_{g\ge 1}\eps^{2g-2}\mu^g\frac{|B_{2g}|}{2(2g)!}uu_{2g}\right)dx.
\end{gather}

Consider case 3. Since $\lambda_g^2=0$, we have
$$
\int_{DR_g(0,0)}\psi_1\lambda_g\lambda_{g-2}=0.
$$
So, there is no quantum correction in this case. Collecting corrections~\eqref{eq:correction1} and~\eqref{eq:correction2} we get the statement of the lemma.
\end{proof}

%%%%%%%%%%%%%%%%%%%%%%%%%%%%%%%%%%%%%%%%%%%%%%%%%%%%%%%%%%%%%%%%%%%%%%%%%%%%%%%%%%%

\subsection{Quantum extended Toda hierarchy}

Consider the cohomological field theory corresponding to the Gromov-Witten theory of~$\CP1$. In~\cite{BR14} we proved that the double ramification hierarchy in this case is equivalent to the extended Toda hierarchy of \cite{CDZ04} by a composition of a certain Miura transformation with a simple triangular transformation. In Section~\ref{subsubsection:quantum for CP1} we compute the Hamiltonian~$\oG_{1,1}$ of the quantum double ramification hierarchy. In Section~\ref{subsubsection:quantum Toda} we discuss how to transform the quantum double ramification hierarchy in order to get a quantization of the extended Toda hierarchy.

We will use our notations from~\cite[Section 6]{BR14} throughout this section.

\subsubsection{Quantum double ramification hierarchy for $\CP1$}\label{subsubsection:quantum for CP1}

\begin{lemma}
We have
\begin{align*}
\oG_{1,1}=\int&\left(\frac{(u^1)^2u^\omega}{2}+\sum_{g\ge 1}\eps^{2g}\frac{B_{2g}}{(2g)!}u^1u^1_{2g}+q\left(\frac{e^{\frac{\eps\d_x}{2}}+e^{-\frac{\eps\d_x}{2}}}{2}u^\omega-2\right)e^{S(\eps\d_x)u^\omega}+q u^\omega\right.\\
&\left.-\frac{i\hbar}{12}u^1+i\hbar\sum_{g\ge 1}\eps^{2g-2}\frac{B_{2g}}{(2g)!}u^\omega_{2g} u^1\right)dx.
\end{align*}
\end{lemma}
\begin{proof}
The classical part $\left.\oG_{1,1}\right|_{\hbar=0}$ was computed in~\cite{BR14}. Let us compute the quantum correction. We have to compute the integrals
\begin{gather}\label{eq:Toda integrals}
\int_{DR_g(0,a_1,\ldots,a_n)}\psi_1\lambda_{g-j} c_{g,n+1,d}(1\otimes\otimes_{i=1}^n\gamma_i),\quad \gamma_i\in\{1,\omega\},
\end{gather}
where $g,j,n\ge 1$.
Denote by $\deg$ the cohomological degree. We have
$$
\deg c_{g,n+1,d}\left(1\otimes\otimes_{i=1}^n\gamma_i\right)=2(g-1-2d)+\sum_{i=1}^n\deg\gamma_i.
$$
Therefore, the integral~\eqref{eq:Toda integrals} is zero unless
\begin{gather*}
2g-2+n=1+g-j+g-1-2d+\frac{1}{2}\sum\deg\gamma_i\quad\Leftrightarrow\quad j=2-2d+\frac{1}{2}\sum\deg\gamma_i-n.
\end{gather*}
Since $\deg\gamma_i\in\{0,2\}$, we immediately conclude that $d=0$ and $j\le 2$. Therefore,
$$
j=2+\frac{1}{2}\sum\deg\gamma_i-n.
$$
We have the following formula:
\begin{gather}\label{eq:degree zero}
c_{g,a+b,0}(1^{\otimes a}\otimes\omega^{\otimes b})=
\begin{cases}
2(-1)^{g-1}\lambda_{g-1},&\text{if $b=0$},\\
(-1)^g\lambda_g,&\text{if $b=1$},\\
0,&\text{otherwise}.
\end{cases}
\end{gather}
This implies that $n\le 2$. We have several cases:
\begin{enumerate}
\item[1.] $j=n=1$.
\item[2.] $j=1$ and $n=2$.
\item[3.] $j=2$.
\end{enumerate}

Consider case 1. We have to compute the integrals
$$
\int_{DR_g(0,0)}\psi_1\lambda_{g-1}c_{g,2,0}\left(1^{\otimes 2}\right)=(-1)^g\int_{\oM_{g,2}}\lambda_g \psi_1\lambda_{g-1}c_{g,2,0}\left(1^{\otimes 2}\right)=-2\int_{\oM_{g,2}}\lambda_g\lambda_{g-1}^2\psi_1.
$$
We have $\lambda_{g-1}^2=2\lambda_{g-2}\lambda_g$, if $g\ge 2$. Therefore, the last integral vanishes, if $g\ge 2$. If $g=1$, then it is equal to~$-\frac{1}{12}$ and the corresponding quantum correction is
$$
-\int\frac{i\hbar}{12}u^1dx.
$$ 

Consider case 2. The corresponding quantum correction is given by
\begin{gather*}
\sum_{g\ge 1}i\hbar(-\eps^2)^{g-1}\sum_{a\in\mbZ}\left(\int_{DR_g(0,a,-a)}\psi_1\lambda_{g-1}(-1)^g\lambda_g\right)p^\omega_a p^1_{-a}=i\hbar\int\sum_{g\ge 1}\eps^{2(g-1)}\frac{B_{2g}}{(2g)!}u^\omega_{2g} u^1 dx.
\end{gather*}

Consider case 3. We get that all $\gamma_i$'s are equal to $\omega$. By~\eqref{eq:degree zero}, the integral~\eqref{eq:Toda integrals} is zero, if $n=2$. Suppose $n=1$, then 
$$
\int_{DR_g(0,0)}\psi_1\lambda_{g-2}c_{g,2,0}\left(1\otimes\omega\right)=\int_{\oM_{g,2}}\lambda^2_g \psi_1\lambda_{g-2}=0.
$$
We see that this case gives no contribution. The lemma is proved.
\end{proof}

\subsubsection{Quantization of the extended Toda hierarchy}\label{subsubsection:quantum Toda}

According to~\cite{BR14} the double ramification hierarchy is related to the extended Toda hierarchy in the following way. Consider a Miura transformation
\begin{gather}\label{eq:Miura for Toda}
v^1(u)=e^{\frac{\eps\d_x}{2}}u^1,\qquad v^2(u)=S(\eps\d_x)u^\omega.
\end{gather}
Denote by $\og_{\alpha,d}[v]$ the Hamiltonians of the double ramification hierarchy rewritten in the jet variables~$v^\alpha_s$. Then the Hamiltonians $\oh^{Td}_{\alpha,p}[v]$ of the extended Toda hierarchy are related to~$\og_{\alpha,p}[v]$ by the following triangular transformation:
\begin{gather}\label{eq:triangular}
\oh^{Td}_{\alpha,p}[v]=\sum_{i=0}^{p+1}(S_i)^\mu_\alpha\og_{\mu,p-i}[v],\quad p\ge -1,
\end{gather}
where the matrices $S_i$ were defined in~\cite[Section 6.1.2]{BR14}. 

We can easily see that the Miura transformation~\eqref{eq:Miura for Toda} naturally induces a map between the Weyl algebra in the variables $p^\alpha_n$ and a Weyl algebra in variables $\tp^\alpha_n$ with a deformed commutator. Indeed, introduce Fourier components of the fields $v^\alpha(x)$:
$$
v^\alpha(x)=\sum_{n\in\mbZ}\tp^\alpha_n e^{inx}.
$$ 
Then the Miura transformation~\eqref{eq:Miura for Toda} induces a map from the Weyl algebra in the variables $p^\alpha_n$ to the Weyl algebra in the variables $\tp^\alpha_n$ by
\begin{gather}\label{eq:quantum Miura}
p^1_n\mapsto \tp^1_n(p)=e^{\frac{in\eps}{2}}p^1_n,\qquad p^\omega_n\mapsto \tp^2_n(p)=S(in\eps)p^\omega_n. 
\end{gather}
The variables $\tp^\alpha_n$ satisfy a deformed commutation relation
$$
[\tp^\alpha_m,\tp^\beta_n]=\hbar\frac{e^{im\eps}-1}{\eps}\delta_{m+n,0}\eta^{\alpha\beta}.
$$ 
If we apply the map~\eqref{eq:quantum Miura} to the quantum Hamiltonians $\oG_{\alpha,p}$ and then compose it with the triangular transformation~\eqref{eq:triangular}, we get quantized Hamiltonians of the extended Toda hierarchy. For example, the quantized Hamiltonian~$\overline H^{Td}_{1,1}[v]$ is equal to
\begin{align*}
\overline H^{Td}_{1,1}[v]=\int&\left(\frac{(v^1)^2}{2}B_-(\eps\d_x)v^2+\left(\frac{\eps\d_x}{2}\coth\left(\frac{\eps\d_x}{2}\right)-1\right)v^1\cdot \left(v^1+\frac{i\hbar}{\eps^2}B_-(\eps\d_x)v^2\right)+\right.\\
&\left.+q\left(\frac{\eps\d_x}{2}\coth\left(\frac{\eps\d_x}{2}\right)v^2-2\right)e^{v^2}-\frac{i\hbar}{12}v^1\right)dx,
\end{align*}
where $B_-(z):=\frac{z}{1-e^{-z}}$.
 
%%%%%%%%%%%%%%%%%%%%%%%%%%%%%%%%%%%%%%%%%%%%%%%%%%%%%%%%%%%%%%%%%%%%%
%%%%%%%%%%%%%%%%%%%%%%%%%%%%%%%%%%%%%%%%%%%%%%%%%%%%%%%%%%%%%%%%%%%%%

{
\appendix

\section{Sums of powers}

In this section we collect some facts about sums of powers. We also prove that a product of derivatives of $\delta_+(x)$ can be expressed as a linear combination of derivatives of $\delta_+(x)$.

\subsection{Sums of powers}\label{subsection:sums of powers}

Let $k\ge 1$ be a positive integer and $d_1,d_2,\ldots,d_k$ be non-negative integers. For any $N\ge 0$ let
$$
\tC^{d_1,d_2,\ldots,d_k}(N):=\sum_{\substack{a_1,\ldots,a_k\in\mbZ_{\ge 0}\\a_1+\ldots+a_k=N}}a_1^{d_1}\ldots a_k^{d_k}.
$$
Here we, by definition, put $0^0:=1$.

\begin{lemma}\label{lemma:sums of powers}
1. The function~$\tC^{d_1,\ldots,d_k}(N)$ is a polynomial in~$N$ with rational coefficients.\\
2. The degree of this polynomial is~$k-1+\sum_{i=1}^k {d_i}$ and the top coefficient is equal to~$\frac{\prod_{i=1}^k d_i!}{(k-1+\sum_{i=1}^k d_i)!}$.
\end{lemma}
\begin{proof}
For $d\ge 0$, the polylogarithm $\Li_{-d}(z)$ is defined by 
$$
\Li_{-d}(z):=\sum_{k\ge 0}k^d z^k.
$$
Note that in the case $d=0$ our definition is slightly different from the standard one. It is easy to see that
$$
\prod_{i=1}^{k}\Li_{-d_i}(z)=\sum_{N\ge 0}\tC^{d_1,\ldots,d_k}(N)z^N.
$$
We have
\begin{gather}\label{eq:formula for the polylogarithm}
\Li_{-d}(z)=\left(z\frac{d}{d z}\right)^d\frac{1}{1-z}=d!\frac{z^d}{(1-z)^{d+1}}+\sum_{i=2}^d a_{i,d}\frac{z^{i-1}}{(1-z)^i},
\end{gather}
where $a_{i,d}$ are some integers. Applying the equation $\frac{z}{1-z}=\frac{1}{1-z}-1$ sufficiently many times, we can express $\Li_{-d}(z)$ in the following way:
$$
\Li_{-d}(z)=\frac{d!}{(1-z)^{d+1}}+\sum_{i=1}^d\frac{b_{i,d}}{(1-z)^i},
$$
for some integers~$b_{i,d}$. Therefore, we have
\begin{gather}\label{eq:product of polylogarithms}
\prod_{i=1}^k\Li_{-d_i}(z)=\frac{\prod_{i=1}^k d_i!}{(1-z)^{k+\sum d_i}}+\sum_{i=1}^{k-1+\sum d_i}\frac{e_{i,d_1,\ldots,d_k}}{(1-z)^i},\qquad e_{i,d_1,\ldots,d_k}\in\mbZ.
\end{gather}
For any $d\ge 1$, we have
$$
\frac{1}{(1-z)^d}=\sum_{i\ge 0}\frac{(i+d-1)(i+d-2)\ldots(i+1)}{(d-1)!}z^i.
$$
From this formula it follows that the coefficient of~$z^N$ on the right-hand side of~\eqref{eq:product of polylogarithms} is a polynomial in~$N$. The second part of the lemma is also clear now.
\end{proof}

Denote by $\tC^{d_1,\ldots,d_k}_j$ the coefficients of the polynomial $\tC^{d_1,\ldots,d_k}(N)$:
\begin{gather}\label{eq:C-polynomial}
\tC^{d_1,\ldots,d_k}(N)=\sum_{j=0}^{k-1+\sum d_i}\tC^{d_1,\ldots,d_k}_j N^j.
\end{gather}
By the previous lemma, the coefficients $\tC^{d_1,\ldots,d_k}_j$ are rational. Note that $\tC^{d_1,\ldots,d_k}_0=0$, if at least one~$d_i$ is non-zero. 
 
\begin{lemma}\label{lemma:parity}
Suppose that $d_1,\ldots,d_k\ge 1$. Then the coefficient $\tC^{d_1,\ldots,d_k}_j$ is equal to zero, if $j\ne k-1+\sum d_i\ (\mathrm{mod}\ 2)$.
\end{lemma}
\begin{proof}
We have
\begin{gather}\label{eq:decomposition}
\prod_{i=1}^k\Li_{-d_i}(z)=\sum_{j=1}^{k-1+\sum d_i}\tC^{d_1,\ldots,d_k}_j\Li_{-j}(z).
\end{gather}
From~\eqref{eq:formula for the polylogarithm} it follows that the polylogarithm~$\Li_{-d}(z)$ is a rational function in~$z$. Moreover, from the middle part of equation~\eqref{eq:formula for the polylogarithm} it follows that
\begin{gather}\label{eq:map z to 1/z}
\Li_{-d}\left(\frac{1}{z}\right)=(-1)^{d+1}\Li_{-d}(z),\quad\text{if $d\ge 1$}.
\end{gather}
Let us replace~$z$ by~$\frac{1}{z}$ in equation~\eqref{eq:decomposition}. Then we get
$$
(-1)^{k+\sum d_i}\prod_{i=1}^k\Li_{-d_i}(z)=\sum_{j=1}^{k-1+\sum d_i}(-1)^{j+1}\tC^{d_1,\ldots,d_k}_j\Li_{-j}(z).
$$
Therefore, we obtain
$$
\sum_{j=1}^{k-1+\sum d_i}\tC^{d_1,\ldots,d_k}_j\Li_{-j}(z)=\sum_{j=1}^{k-1+\sum d_i}(-1)^{k-1+\sum d_i-j}\tC^{d_1,\ldots,d_k}_j\Li_{-j}(z).
$$
From~\eqref{eq:formula for the polylogarithm} it also follows that the functions~$\Li_{-d}(z), d\ge 0$, are linearly independent. We conclude that the coefficient $\tC^{d_1,\ldots,d_k}_j$ can be non-zero only if $j=k-1+\sum d_i\ (\mathrm{mod}\ 2)$.
\end{proof}

Now we want to prove an auxiliary lemma that we will use in Appendix~\ref{section:polynomiality}.
Let $k\ge 1$ and $n\ge 0$. Consider variables $a_1,\ldots,a_n$, $b_1,\ldots,b_k$ and $b$. Let
$$
A:=(a_1,\ldots,a_n),\quad\text{and}\quad B:=(b_1,\ldots,b_k).
$$
For any polynomial $P(A,b,B)\in\mbC[a_1,\ldots,a_n,b,b_1,\ldots,b_k]$, define the function $S_{b,B}[P](A,N)$ by
\begin{gather*}
S_{b,B}[P](A,N):=\sum_{\substack{b,b_1,\ldots,b_k\in\mbZ_{\ge 1}\\b+\sum b_i=N}}P(A,b,B)\prod_{i=1}^k b_i+\frac{1}{2}\sum_{\substack{b_1,\ldots,b_k\in\mbZ_{\ge 1}\\\sum b_i=N}}P(A,0,B)\prod_{i=1}^k b_i,\quad N\ge 0.
\end{gather*}
From Lemma~\ref{lemma:sums of powers} it follows that the function~$S_{b,B}[P](A,N)$ is a polynomial in~$a_1,\ldots,a_n$ and $N$ of degree not more than~$2k+\deg P$. Note that
\begin{gather}\label{eq:zero of S}
S_{b,B}(A,0)=0.
\end{gather}

\begin{lemma}\label{lemma:lemma for S}
Suppose a polynomial $P(A,b,B)$ is even. Then the polynomial~$S_{b,B}[P](A,N)$ is also even.
\end{lemma}
\begin{proof}
By linearity, it is sufficient to prove the lemma when $P$ is a monomial:
$$
P=\left(\prod_{i=1}^n a_i^{r_i}\right)b^d\left(\prod_{i=1}^kb_i^{d_i}\right).
$$
If $d\ge 1$, then the lemma follows from Lemma~\ref{lemma:parity}. Suppose $d=0$. Then we proceed in the same way as in the proof of Lemma~\ref{lemma:parity}. Let $D:=\sum_{i=1}^k d_i$. We have 
$$
\frac{S_{b,B}[P](A,N)}{\prod a_i^{r_i}}=\sum_{j=1}^{D+2k}c_j N^j,\quad c_j\in\mathbb Q.
$$
Therefore,
$$
\frac{1}{\prod a_i^{r_i}}\sum_{N\ge 0}S_{b,B}[P](A,N)z^N=\sum_{j=1}^{D+2k}c_j\Li_{-j}(z).
$$
From the definition of $S_{b,B}[P](A,N)$ it follows that
\begin{gather}\label{eq:formula with S_P}
\frac{1}{\prod a_i^{r_i}}\sum_{N\ge 0}S_{b,B}[P](A,N)z^N=\frac{1}{2}\frac{1+z}{1-z}\prod_{i=1}^k\Li_{-d_i-1}(z).
\end{gather}
Equation~\eqref{eq:map z to 1/z} implies that the right-hand side of~\eqref{eq:formula with S_P} is multiplied by $(-1)^{D+1}$ under the map $z\mapsto\frac{1}{z}$. Therefore, the coefficient~$c_j$ is zero unless $j=D\ (\mathrm{mod}\ 2)$. The lemma is proved. 
\end{proof}

\subsection{Products of derivatives of $\delta_+(x)$}\label{subsection:products}

Recall that 
$$
\delta^{(s)}_+(x)=\sum_{k\ge 0}(ik)^s e^{ikx},\quad s\ge 0.
$$
\begin{lemma}\label{lemma:products of derivatives}
Suppose $n\ge 1$ and $a_1,\ldots,a_n\ge 1$. The product $\prod_{i=1}^n\delta_+^{(a_i)}(x)$ can be expressed in the following way:
$$
\prod_{i=1}^n\delta_+^{(a_i)}(x)=(-i)^{n-1}\sum_{j=1}^{n-1+\sum a_i}C^{a_1,\ldots,a_n}_j\delta_+^{(j)},
$$
where 
\begin{gather}\label{eq:relation of coefficients}
C_j^{a_1,\ldots,a_n}=
\begin{cases}
(-1)^{\frac{n-1+\sum a_i-j}{2}}\tC_j^{a_1,\ldots,a_n},&\text{if $j=n-1+\sum_{i=1}^n a_i\ (\mathrm{mod}\ 2)$},\\
0,&\text{otherwise}.
\end{cases}
\end{gather}
\end{lemma}
\begin{proof}
We have
\begin{multline*}
\prod_{i=1}^n\delta_+^{(a_i)}(x)=i^{\sum a_i}\sum_{N\ge 0}\tC^{a_1,\ldots,a_n}(N)e^{iNx}\stackrel{\text{by eq.~\eqref{eq:C-polynomial}}}{=}\sum_{j=1}^{n-1+\sum a_i}i^{\sum a_i-j}\tC^{a_1,\ldots,a_n}_j\delta^{(j)}_+(x)=\\
=(-i)^{n-1}\sum_{j=1}^{n-1+\sum a_i}i^{n-1+\sum a_i-j}\tC^{a_1,\ldots,a_n}_j\delta^{(j)}_+(x).
\end{multline*}
By Lemma~\ref{lemma:parity}, the coefficient $\tC^{a_1,\ldots,a_n}_j$ can be non-zero only if $j=n-1+\sum_{i=1}^n a_i\ (\mathrm{mod}\ 2)$. The lemma is proved.
\end{proof}

%%%%%%%%%%%%%%%%%%%%%%%%%%%%%%%%%%%%%%%%%%%%%%%%%%%%%%%%%%%%%%%%%%%%%
%%%%%%%%%%%%%%%%%%%%%%%%%%%%%%%%%%%%%%%%%%%%%%%%%%%%%%%%%%%%%%%%%%%%%

\section{Polynomiality of tautological integrals}\label{section:polynomiality}

In this section we prove that the integral of an arbitrary tautological class over the double ramification cycle $DR_g(a_1,\ldots,a_n)$ is a polynomial in the ramification multiplicities $a_1,\ldots,a_n$. In Section~\ref{subsection:tautological classes} we briefly recall the notion of the tautological ring in the cohomology of the moduli space of curves. Section~\ref{subsection:polynomiality} is devoted to the proof of the polynomiality statement.

\subsection{Tautological classes}\label{subsection:tautological classes}

The system of tautological rings is defined to be the set of the smallest $\mbC$-subalgebras of the cohomology rings,
$$
RH^*(\oM_{g,n})\subset H^*(\oM_{g,n};\mbC),
$$
satisfying the following two properties:
\begin{enumerate}

\item The system is closed under push-forward via all maps forgetting markings:
$$
\pi_*\colon RH^*(\oM_{g,n})\to RH^*(\oM_{g,n-1}).
$$

\item The system is closed under push-forward via all gluing maps:
\begin{align*}
&gl_*\colon RH^*(\oM_{g_1,n_1+1})\otimes_{\mbC}RH^*(\oM_{g_2,n_2+1})\to RH^*(\oM_{g_1+g_2,n_1+n_2}),\\
&gl_*\colon RH^*(\oM_{g-1,n+2})\to RH^*(\oM_{g,n}).
\end{align*}
\end{enumerate}
While the definition appears restrictive, natural algebraic constructions typically yield cohomology classes lying in the tautological ring. For example, boundary classes, the standard $\psi$, $\kappa$, and $\lambda$ classes all lie in the tautological ring.

The tautological ring $RH^*(\oM_{g,n})$ can be also described as an image of a certain finite-dimensional $\mbC$-algebra $\mcS_{g,n}$, called the strata algebra. We recommend the reader the paper \cite[Sections 0.2, 0.3]{PPZ15} as a very good introduction in this subject. 

%%%%%%%%%%%%%%%%%%%%%%%%%%%%%%%%%%%%%%%%%%%%%%%%%%%%%%%%%%%%%%%%%%%%%%%%%%%%%%%%%

\subsection{Polinomiality of integrals}\label{subsection:polynomiality}

Recall that the double ramification cycle $DR_g(a_1,\ldots,a_n)$ is defined for arbitrary integers $a_1,\ldots,a_n$ such that $a_1+\ldots+a_n=0$. We remind the reader the following two properties:
\begin{align}
&DR_g(-a_1,\ldots,-a_n)=DR_g(a_1,\ldots,a_n),\label{eq:sign change}\\
&DR_g(0,\ldots,0)=(-1)^g\lambda_g.\notag
\end{align}

\begin{proposition}
Let $n\ge 1$. For an arbitrary tautological class $\alpha\in RH^*(\oM_{g,n})$ the integral
$$
\int_{DR_g(a_1,\ldots,a_n)}\alpha
$$
is a polynomial in~$a_1,\ldots,a_n$. Moreover, this polynomial is even and its degree is less or equal to~$2g$. 
\end{proposition}
\begin{proof}
There is nothing to prove for $n=1$. Suppose $n\ge 2$. We will use the notations from \cite[Sections 0.2, 0.3]{PPZ15}. The space $\mcS_{g,n}$ has a basis whose elements are the isomorphism classes of pairs $[\Gamma,\gamma]$, where
$$
\Gamma=(V,H,L,g\colon V\to\mbZ_{\ge 0},v\colon H\to V,\iota\colon H\to H)
$$ 
is a stable graph of genus $g$ with $n$ legs and 
\begin{gather}\label{eq:gamma}
\gamma=\prod_{v\in V}\prod_{i\ge 1}\kappa_i[v]^{x_i[v]}\cdot\prod_{h\in H}\psi_h^{y[h]}\in H^*(\oM_{\Gamma};\mbC)
\end{gather}
is a basic class on 
$$
\oM_{\Gamma}:=\prod_{v\in V}\oM_{g(v),n(v)}.
$$
There is a canonical morphism 
$$
\xi_\Gamma\colon\oM_{\Gamma}\to\oM_{g,n}.
$$
The map $q\colon\mcS_{g,n}\to H^*(\oM_{g,n})$ is defined by
$$
q([\Gamma,\gamma]):=\xi_{\Gamma*}(\gamma).
$$
The tautological ring $RH^*(\oM_{g,n})$ coincides with the image of $q$.

Without loss of generality we can assume that $\alpha=q([\Gamma,\gamma])$. We proceed by induction on the number of edges in the graph~$\Gamma$. Suppose it has no edges, then $\oM_\Gamma=\oM_{g,n}$ and
\begin{gather}\label{eq:alpha}
\alpha=\gamma=\prod_{i\ge 1}\kappa_i^{x_i}\cdot\prod_{i=1}^n\psi_i^{y_i}.
\end{gather}
It is well-known (see e.g.~\cite[Section 2.1]{Ion02}) that the class~\eqref{eq:alpha} can be expressed as a linear combination of classes of the form
$$
\pi_{m*}\left(\prod_{i=1}^{n+m}\psi_i^{d_i}\right),
$$
for some $m$'s, where $\pi_m\colon\oM_{g,n+m}\to\oM_{g,n}$ is the map that forgets the last $m$ marked points.
We have
$$
\int_{DR_g(a_1,\ldots,a_n)}\pi_{m*}\left(\prod_{i=1}^{n+m}\psi_i^{d_i}\right)=\int_{DR_g(a_1,\ldots,a_n,0,\ldots,0)}\prod_{i=1}^{n+m}\psi_i^{d_i}.
$$
In \cite{BSSZ12} it is proved that an integral of an arbitrary monomial in psi-classes over the double ramification cycle $DR_g(b_1,\ldots,b_l)$ is an even polynomial in $b_1,\ldots,b_l$ of degree not more than $2g$. Therefore, the integral $\int_{DR_g(a_1,\ldots,a_n)}\alpha$ is an even polynomial in $a_1,\ldots,a_n$ of degree not greater than~$2g$.

Suppose the graph $\Gamma$ has $l$ edges and $l\ge 1$. Let us choose some edge $e$. Suppose $e$ is separating, then if we cut it in two half-edges, then the graph $\Gamma$ will become a disjoint union of two stable graphs. Denote them by~$\Gamma_1$ and~$\Gamma_2$. Let $V_i, H_i$ and $L_i$ be the set of vertices, the set of half-edges and the set of legs of the graph $\Gamma_i$ correspondingly. Denote by $g_i$ the genus of the graph $\Gamma_i$. Let 
$$
\gamma_k:=\prod_{v\in V_k}\prod_{i\ge 1}\kappa_i[v]^{x_i[v]}\cdot\prod_{h\in H_k}\psi_h^{y[h]}\in H^*(\oM_{\Gamma_k};\mbC),\quad k=1,2.
$$
Denote by $gl\colon\oM_{g_1,|L_1|+1}\times\oM_{g_2,|L_2|+1}\to\oM_{g,n}$ the gluing morphism. We see that
$$
q([\Gamma,\gamma])=\xi_{\Gamma*}(\gamma)=gl_*\left(\xi_{\Gamma_1*}(\gamma_1)\otimes\xi_{\Gamma_2*}(\gamma_2)\right).
$$
Recall that for a subset $I=\{i_1,i_2,\ldots\},i_1<i_2<\ldots$, of $\{1,2,\ldots,n\}$ we use the following notations:
\begin{gather*}
A_I=(a_{i_1},a_{i_2},\ldots),\qquad a_I=\sum_{i\in I}a_i.
\end{gather*}
We have (see \cite{BSSZ12}) 
$$
gl^* DR_g(a_1,\ldots,a_n)=DR_{g_1}(A_{L_1},-a_{L_1})\otimes DR_{g_2}(A_{L_2},-a_{L_2}).
$$
Thus, we get
\begin{gather}\label{eq:factors}
\int_{DR_g(a_1,\ldots,a_n)}q([\Gamma,\gamma])=\left(\int_{DR_{g_1}(A_{L_1},-a_{L_1})}\xi_{\Gamma_1*}(\gamma_1)\right)\left(\int_{DR_{g_2}(A_{L_2},-a_{L_2})}\xi_{\Gamma_2*}(\gamma_2)\right).
\end{gather}
The number of edges in each graph $\Gamma_i$ is less than $l$. Therefore, by induction assumption, both factors on the right-hand side of~\eqref{eq:factors} are even polynomials in~$a_i$'s of degrees not more than~$2g_1$ and~$2g_2$ correspondingly. Therefore, the integral on the right-hand side of~\eqref{eq:factors} is an even polynomial in $a_i$'s of degree not more than $2(g_1+g_2)=2g$.  

Suppose that the edge~$e$ is non-separating. Denote by $\tGamma$ the stable graph obtained from $\Gamma$ by cutting the edge~$e$ in two half-edges. We mark two new legs by~$n+1$ and~$n+2$ correspondingly. The space~$\oM_{\tGamma}$ is naturally isomorphic to~$\oM_{\Gamma}$. Let $gl\colon\oM_{g-1,n+2}\to\oM_{g,n}$ be the gluing morphism. We can decompose the map $\xi_\Gamma\colon\oM_{\Gamma}\to\oM_{g,n}$ in the following way:
\begin{center}
\begin{tikzcd}
\oM_{\Gamma}\arrow{r}{\sim}\arrow[bend right]{rrr}{\xi_\Gamma} & \oM_{\tGamma}\arrow{r}{\xi_{\tGamma}} &\oM_{g-1,n+2}\arrow{r}{gl} & \oM_{g,n}
\end{tikzcd}
\end{center}
Therefore, we have
$$
q([\Gamma,\gamma])=\xi_{\Gamma*}(\gamma)=gl_*(\xi_{\tGamma*}(\tgamma)),
$$
where $\tgamma$ is the cohomology class on $\oM_{\tGamma}$ induced from $\gamma$ by the isomorphism $\oM_{\tGamma}\stackrel{\sim}{\to}\oM_\Gamma$. Therefore, 
\begin{gather}\label{eq:second integral}
\int_{DR_g(a_1,\ldots,a_n)}q([\Gamma,\gamma])=\int_{gl^*(DR_g(a_1,\ldots,a_n))}\xi_{\tGamma*}(\tgamma).
\end{gather}
Denote by $[n]$ the set $\{1,2,\ldots,n\}$. We have (\cite{Zvo})
\begin{align}
gl^*DR_g&(a_1,\ldots,a_n)=\label{eq:pullback}\\
=&\sum_{\substack{I\sqcup J=[n]\\a_I>0}}\sum_{k\ge 1}\sum_{\substack{g_1,g_2\ge 0\\g_1+g_2+k=g}}\sum_{\substack{B=(b_1,\ldots,b_k)\in\mbZ_{\ge 1}^k\\b\ge 1,\,b+\sum b_i=a_I}}\frac{\prod_{i=1}^kb_i}{k!}DR_{g_1}(A_I,\stackrel{n+1}{\boxed{-b}},-B)\boxtimes DR_{g_2}(A_J,\stackrel{n+2}{\boxed{b}},B)+\label{eq:1}\\
&+\sum_{\substack{I\sqcup J=[n]\\a_I>0}}\sum_{k\ge 1}\sum_{\substack{g_1,g_2\ge 0\\g_1+g_2+k=g}}\sum_{\substack{B=(b_1,\ldots,b_k)\in\mbZ_{\ge 1}^k\\b\ge 1,\,b+\sum b_i=a_I}}\frac{\prod_{i=1}^k b_i}{k!}DR_{g_1}(A_I,\stackrel{n+2}{\boxed{-b}},-B)\boxtimes DR_{g_2}(A_J,\stackrel{n+1}{\boxed{b}},B)+\label{eq:2}\\
&+\frac{1}{2}\sum_{\substack{I\sqcup J=[n]\\a_I>0}}\sum_{k\ge 1}\sum_{\substack{g_1,g_2\ge 0\\g_1+g_2+k=g}}\sum_{\substack{B=(b_1,\ldots,b_k)\in\mbZ_{\ge 1}^k\\\sum b_i=a_I}}\frac{\prod_{i=1}^k b_i}{k!}DR_{g_1}(A_I,\stackrel{n+1}{\boxed{0}},-B)\boxtimes DR_{g_2}(A_J,\stackrel{n+2}{\boxed{0}},B)+\label{eq:3}\\
&+\frac{1}{2}\sum_{\substack{I\sqcup J=[n]\\a_I>0}}\sum_{k\ge 1}\sum_{\substack{g_1,g_2\ge 0\\g_1+g_2+k=g}}\sum_{\substack{B=(b_1,\ldots,b_k)\in\mbZ_{\ge 1}^k\\\sum b_i=a_I}}\frac{\prod_{i=1}^k b_i}{k!}DR_{g_1}(A_I,\stackrel{n+2}{\boxed{0}},-B)\boxtimes DR_{g_2}(A_J,\stackrel{n+1}{\boxed{0}},B).\label{eq:4}
\end{align}
Here the notation $\stackrel{n+1}{\boxed{-b}}$ means that the point marked by $n+1$ has ramification multiplicity $-b$. We are grateful to D.~Zvonkine for informing us about this formula.

Consider an arbitrary decomposition $I\sqcup J=[n]$ and numbers $g_1,g_2\ge 0$, $k\ge 1$, such that $g_1+g_2+k=g$. Denote by $\Phi_{g_1,g_2,k,I,J}$ the stable graph described as follows:
\begin{itemize}
\item It has two vertices with $k$ edges between them; 
\item The genus of the first vertex is $g_1$ and the genus of the second one is $g_2$;
\item The first vertex contains the legs marked by $I$ and also the leg marked by $n+1$. The second vertex contains the legs marked by $J$ and also the leg marked by $n+2$.
\end{itemize} 
We see that each term on the right-hand side of formula~\eqref{eq:pullback} has the form
$$
\xi_{\Phi_{g_1,g_2,k,I,J}*}(\beta_1\otimes\beta_2),
$$
where $\beta_1$ and $\beta_2$ are double ramification cycles on~$\oM_{g_1,|I|+k+1}$ and~$\oM_{g_2,|J|+k+1}$ correspondingly. From formulas in~\cite[Section 0.3]{PPZ15} it follows that the class 
$$
\xi^*_{\Phi_{g_1,g_2,k,I,J}}\left(\xi_{\tGamma*}(\tgamma)\right)\in H^*(\oM_{\Phi_{g_1,g_2,k,I,J}};\mbC)
$$
in the cohomology of 
$$
\oM_{\Phi_{g_1,g_2,k,I,J}}=\oM_{g_1,|I|+k+1}\times\oM_{g_2,|J|+k+1}
$$
can be expressed in the following way:
$$
\xi^*_{\Phi_{g_1,g_2,k,I,J}}\left(\xi_{\tGamma *}(\tgamma)\right)=\sum_j e_j\left(\xi_{\tGamma_{j,1}*}(\tgamma_{j,1})\right)\otimes\left(\xi_{\tGamma_{j,2}*}(\tgamma_{j,2})\right),
$$
where $\tGamma_{j,i}$ are stable graphs, $\tgamma_{j,i}$ are basic classes and $e_j$ are some rational coefficients. Of course,~$\tGamma_{j,i}, \tgamma_{j,i}$ and~$e_j$ depend on $g_1,g_2,k,I,J$, but in order to shorten the exposition of the paper, we omit it in their notations. Most importantly, from the procedure, described in~\cite{PPZ15}, it follows that the number of edges the graphs~$\tGamma_{j,i}$ is not more than the number of edges in the graph $\tGamma$, so it is not more than~$l-1$. From the induction assumption it follows that the function
$$
\frac{1}{k!}\sum_j e_j\left(\int_{DR_{g_1}(A_I,\stackrel{\scriptscriptstyle n+1}{\boxed{\scriptstyle-b}},-B)}\xi_{\tGamma_{j,1}*}\left(\tgamma_{j,1}\right)\right)\left(\int_{DR_{g_2}(A_J,\stackrel{\scriptscriptstyle n+2}{\boxed{\scriptstyle b}},B)}\xi_{\tGamma_{j,2}*}\left(\tgamma_{j,2}\right)\right)
$$
is an even polynomial in $a_1,\ldots,a_n,b,b_1,\ldots,b_k$ of degree not greater than~$2(g_1+g_2)$. Denote it by 
$$
P_{g_1,g_2,k,I,J}(A,b,B).
$$

Let us prove that the integral~\eqref{eq:second integral} is equal to 
\begin{gather}\label{eq:resulting polynomial}
\sum_{I\sqcup J=[n]}\sum_{k\ge 1}\sum_{\substack{g_1,g_2\ge 0\\g_1+g_2+k=g}}S_{b,B}[P_{g_1,g_2,k,I,J}](A,a_I).
\end{gather}
Consider the part of the integral~\eqref{eq:second integral}, that corresponds to the term~\eqref{eq:1} in the expression for $gl^*(DR_g(a_1,\ldots,a_n))$:
\begin{align}
&\sum_{\substack{I\sqcup J=[n]\\a_I>0}}\sum_{k\ge 1}\sum_{\substack{g_1,g_2\ge 0\\g_1+g_2+k=g}}\sum_{\substack{B=(b_1,\ldots,b_k)\in\mbZ_{\ge 1}^k\\b\ge 1,\,b+\sum b_i=a_I}}\frac{\prod_{i=1}^k b_i}{k!}\int_{DR_{g_1}(A_I,\stackrel{\scriptscriptstyle n+1}{\boxed{\scriptstyle-b}},-B)\boxtimes DR_{g_2}(A_J,\stackrel{\scriptscriptstyle n+2}{\boxed{\scriptstyle b}},B)}\xi_{\tGamma*}(\tgamma)=\notag\\
=&\sum_{\substack{I\sqcup J=[n]\\a_I>0}}\sum_{k\ge 1}\sum_{\substack{g_1,g_2\ge 0\\g_1+g_2+k=g}}\sum_{\substack{B=(b_1,\ldots,b_k)\in\mbZ_{\ge 1}^k\\b\ge 1,\,b+\sum b_i=a_I}}\frac{\prod_{i=1}^k b_i}{k!}\sum_j e_j\times\notag\\
&\times\left(\int_{DR_{g_1}(A_I,\stackrel{\scriptscriptstyle n+1}{\boxed{\scriptstyle-b}},-B)}\xi_{\tGamma_{j,1}*}(\tgamma_{j,1})\right)\left(\int_{DR_{g_2}(A_J,\stackrel{\scriptscriptstyle n+2}{\boxed{\scriptstyle b}},B)}\xi_{\tGamma_{j,2}*}(\tgamma_{j,2})\right)=\notag\\
=&\sum_{\substack{I\sqcup J=[n]\\a_I>0}}\sum_{k\ge 1}\sum_{\substack{g_1,g_2\ge 0\\g_1+g_2+k=g}}\sum_{\substack{B=(b_1,\ldots,b_k)\in\mbZ_{\ge 1}^k\\b\ge 1,\,b+\sum b_i=a_I}}\left(\prod_{i=1}^k b_i\right) P_{g_1,g_2,k,I,J}(A,b,B).\label{eq:comp}
\end{align}
Doing the same computation with the term~\eqref{eq:3} and adding it to~\eqref{eq:comp}, we get
\begin{gather}\label{eq:first summand}
\sum_{\substack{I\sqcup J=[n]\\a_I>0}}\sum_{k\ge 1}\sum_{\substack{g_1,g_2\ge 0\\g_1+g_2+k=g}}S_{b,B}[P_{g_1,g_2,k,I,J}](A,a_I).
\end{gather}
Let us consider the part of the integral~\eqref{eq:second integral}, that corresponds to the term~\eqref{eq:2}. We obtain:
\begin{align}
&\sum_{\substack{I\sqcup J=[n]\\a_I<0}}\sum_{k\ge 1}\sum_{\substack{g_1,g_2\ge 0\\g_1+g_2+k=g}}\sum_{\substack{B=(b_1,\ldots,b_k)\in\mbZ_{\ge 1}^k\\b\ge 1,\,b+\sum b_i=-a_I}}\frac{\prod_{i=1}^k b_i}{k!}\sum_j e_j\times\notag\\
&\times\left(\int_{DR_{g_1}(A_I,\stackrel{\scriptscriptstyle n+1}{\boxed{\scriptstyle b}},B)}\xi_{\tGamma_{j,1}*}(\tgamma_{j,1})\right)\left(\int_{DR_{g_2}(A_J,\stackrel{\scriptscriptstyle n+2}{\boxed{\scriptstyle -b}},-B)}\xi_{\tGamma_{j,2}*}(\tgamma_{j,2})\right)\stackrel{\text{by eq. \eqref{eq:sign change}}}{=}\notag\\
=&\sum_{\substack{I\sqcup J=[n]\\a_I<0}}\sum_{k\ge 1}\sum_{\substack{g_1,g_2\ge 0\\g_1+g_2+k=g}}\sum_{\substack{B=(b_1,\ldots,b_k)\in\mbZ_{\ge 1}^k\\b\ge 1,\,b+\sum b_i=-a_I}}\frac{\prod_{i=1}^k b_i}{k!}\sum_j e_j\times\notag\\
&\times\left(\int_{DR_{g_1}(-A_I,\stackrel{\scriptscriptstyle n+1}{\boxed{\scriptstyle -b}},-B)}\xi_{\tGamma_{j,1}*}(\tgamma_{j,1})\right)\left(\int_{DR_{g_2}(-A_J,\stackrel{\scriptscriptstyle n+2}{\boxed{\scriptstyle b}},B)}\xi_{\tGamma_{j,2}*}(\tgamma_{j,2})\right)=\notag\\
=&\sum_{\substack{I\sqcup J=[n]\\a_I<0}}\sum_{k\ge 1}\sum_{\substack{g_1,g_2\ge 0\\g_1+g_2+k=g}}\sum_{\substack{B=(b_1,\ldots,b_k)\in\mbZ_{\ge 1}^k\\b\ge 1,\,b+\sum b_i=-a_I}}\left(\prod_{i=1}^k b_i\right)P_{g_1,g_2,k,I,J}(-A,b,B).\label{eq:comp-}
\end{align}
Doing the same computation for~\eqref{eq:4} and adding it to~\eqref{eq:comp-} we get
\begin{gather*}
\sum_{\substack{I\sqcup J=[n]\\a_I<0}}\sum_{k\ge 1}\sum_{\substack{g_1,g_2\ge 0\\g_1+g_2+k=g}}S_{b,B}[P_{g_1,g_2,k,I,J}](-A,-a_I)\stackrel{\text{by Lemma~\ref{lemma:lemma for S}}}{=}\sum_{\substack{I\sqcup J=[n]\\a_I<0}}\sum_{k\ge 1}\sum_{\substack{g_1,g_2\ge 0\\g_1+g_2+k=g}}S_{b,B}[P_{g_1,g_2,k,I,J}](A,a_I).
\end{gather*}
Summing the last expression with~\eqref{eq:first summand} and using~\eqref{eq:zero of S}, we get~\eqref{eq:resulting polynomial}. By Lemma~\ref{lemma:lemma for S}, the sum~\eqref{eq:resulting polynomial} is an even polynomial in $a_1,\ldots,a_n$ of degree at most~$2g$. The proposition is proved.
\end{proof}

}

\end{document}